\documentclass[11pt]{article}

\usepackage{amsfonts,amsmath,amssymb, bbm, bm, xspace, mathrsfs}
\usepackage{multirow,graphicx, epsfig,subfigure,algorithm,algorithmic,rotating}
\usepackage{cite,url}
\usepackage{fullpage}
\usepackage[small,bf]{caption}
\usepackage[table]{xcolor}
\usepackage{tabularx}
\usepackage{psfrag}
\usepackage{subfigure}
\usepackage{verbatim}
\newtheorem{theorem}{Theorem}
\newtheorem{proposition}[theorem]{Proposition}

\newtheorem{lemma}[theorem]{Lemma}

\newtheorem{definition}[theorem]{Definition}
\newtheorem{example}[theorem]{Example}

\def \endprf{\hfill {\vrule height6pt width6pt depth0pt}\medskip}
\newenvironment{proof}{\noindent {\bf Proof} }{\endprf\par}

\numberwithin{equation}{section}

\newcommand{\codebook}{\mathcal{C}}

\newcommand{\argmax}{\operatornamewithlimits{argmax}}
\newcommand{\conv}{\operatorname{conv}}
\newcommand{\vect}{\operatorname{vec}}
\newcommand{\st}{\operatorname{subject~to}}

\newcommand{\argmin}{\operatornamewithlimits{argmin}}

\newcommand{\mbf}{\boldsymbol}
\newcommand{\mc}{\mathcal}
\newcommand{\real}{\mathbb{R}}
\newcommand{\integer}{\mathbb{Z}}

\newcommand{\bfa}{\mbf{a}}

\newcommand{\bfe}{\mbf{e}}
\newcommand{\bfx}{\mbf{x}}
\newcommand{\bfu}{\mbf{u}}

\newcommand{\bfq}{\mbf{q}}
\newcommand{\bfr}{\mbf{r}}
\newcommand{\bfs}{\mbf{s}}
\newcommand{\bft}{\mbf{t}}
\newcommand{\bfy}{\mbf{y}}
\newcommand{\bfA}{\mbf{A}}
\newcommand{\bfb}{\mbf{b}}
\newcommand{\bfc}{\mbf{c}}

\newcommand{\bfB}{\mbf{B}}
\newcommand{\bfE}{\mbf{E}}

\newcommand{\bfP}{\mbf{P}}
\newcommand{\bfQ}{\mbf{Q}}
\newcommand{\bfY}{\mbf{Y}}
\newcommand{\bfX}{\mbf{X}}
\newcommand{\bfZ}{\mbf{Z}}

\newcommand{\bfzero}{\mbf{0}}
\newcommand{\bfone}{\mbf{1}}
\newcommand{\bfPi}{\mbf{P}}
\newcommand{\mcI}{\mc{I}}

\newcommand{\tr}{\mathrm{tr}}

\newcommand{\symgroup}{\mathbb{S}}

\newcommand{\kendall}{d_K}
\newcommand{\hamming}{d_H}
\newcommand{\chebyshev}{d_{\infty}}

\newcommand{\mpmX}{\bfX}
\newcommand{\mpmY}{\bfY}
\newcommand{\mpmXent}{X}
\newcommand{\mpmYent}{Y}

\newcommand{\mpmset}{\mathcal{M}}
\newcommand{\mpx}{\mbf{x}}
\newcommand{\mpxent}{x}
\newcommand{\mpy}{\mbf{y}}
\newcommand{\mpyent}{y}
\newcommand{\multipermpolytope}{\mathbb{M}}
\newcommand{\mpcp}{\mathbb{P}^{\mathsf{M}}}
\newcommand{\mpset}{\mathsf{M}}

\newcommand{\LCMM}{\Pi^{\mathsf{M}}} 
\newcommand{\LCMC}{\Lambda^{\mathsf{M}}} 

\newcommand{\bftelta}{\mbf{\delta}}
\newcommand{\LLRent}{\gamma}
\newcommand{\LLRvect}{\mbf{\gamma}}
\newcommand{\LLRbig}{\mbf{\Gamma}}
\newcommand{\cin}{\mathcal{S}}
\newcommand{\cout}{\Sigma}
\renewcommand{\Pr}{P}
\renewcommand{\min}{\operatorname{minimize}}
\newcommand{\nequiv}{\not\equiv}
\newcommand\bigzero{\makebox(0,0){\text{\huge0}}}

\DeclareMathOperator*{\eq}{=}

\title{LP-decodable multipermutation codes\thanks{This work was supported by the National Science Foundation (NSF) under Grants CCF-1217058 and by a Natural Sciences and Engineering Research Council of Canada (NSERC) Discovery Research Grant. To appear at the \emph{2014 Allerton Conference on Communication, Control, and Computing.}}}

\begin{document}

\author{Xishuo Liu
\thanks{X.~Liu is with the Dept.~of Electrical and Computer
  Engineering, University of Wisconsin, Madison, WI 53706
  (e-mail: xliu94@wisc.edu).}  
, Stark C. Draper
\thanks{S.~C.~Draper is with the Dept.~of Electrical and Computer
  Engineering, University of Toronto, ON M5S 3G4, Canada (e-mail: stark.draper@utoronto.ca).}
}

\maketitle

\begin{abstract}
In this paper, we introduce a new way of constructing and decoding multipermutation codes. Multipermutations are permutations of a multiset that may consist of duplicate entries. We first introduce a new class of matrices called multipermutation matrices. We characterize the convex hull of multipermutation matrices. Based on this characterization, we propose a new class of codes that we term LP-decodable multipermutation codes. Then, we derive two LP decoding algorithms. We first formulate an LP decoding problem for memoryless channels. We then derive an LP algorithm that minimizes the Chebyshev distance. Finally, we show a numerical example of our algorithm.
\end{abstract}

\section{Introduction}
Using permutations and multipermutations in communication systems dates back to~\cite{slepian1965permutation}, where Slepian considered using multipermutations as a data transmission scheme in the additive white Gaussian noise (AWGN) channel. In recent years, there has been a growing interest in permutation codes due to their applications in various areas such as power line communications (PLC)~\cite{chu2004constructions} and flash memories~\cite{jiang2008rank}. For PLC, permutation codes are proposed to deal with permanent narrow-band noise and impulse noises while delivering constant transmission power (see also~\cite{colbourn2004permutation}). For flash memories, information is stored using the charge level of memory cells. Jiang \emph{et al.} proposed using the relative rankings among cells to modulate information~\cite{jiang2008rank}. This approach is able to reduce the errors caused by charge leakage. Further, it alleviates the over-injection problem during cell programming (cf.~\cite{jiang2008rank}). 

In this paper, we consider coding using multipermutations. Loosely speaking, multipermutations are permutations that contain duplicates. Each multipermutation is a permutation of the multiset $\{1,1,\dots,1,2,\dots,2,\dots,m,\dots,m\}$. We associate each set of multipermutations with a multiplicity vector $\bfr = (r_1,\dots,r_m)$, where $r_i$ is the number of entries with value $i$ in the multiset. In the literature, codes using multipermutations are referred as constant-composition codes
when the Hamming distance is considered~\cite{chu2006on}. When $r_1 = r_2 = \dots = r_m$, the multipermutations under consideration are known as frequency permutation arrays~\cite{huczynska2006frequency}. Recently, multipermutation codes under the Kendall tau distance and the Ulam distance are studied in~\cite{buzaglo2013error} and~\cite{farnoud2014multipermutation} respectively. There are two motivations for coding using multipermutations: First, the size of the codebook based on multipermutations can be larger than that based on (multiple) permutations. Second, the number of distinct charges a flash memory can store is limited by the physical resolution of the hardware, which means that using permutations over large alphabets is impractical.  

LP-decodable permutation codes are proposed by Wadayama and Hagirawa in~\cite{wadayama2012lpdecodable}. Indeed, the construction in~\cite{wadayama2012lpdecodable} is already defined over multipermutations. However, there are two issues to this construction. First, multipermutations are described using permutation matrices in~\cite{wadayama2012lpdecodable}. As a result, the number of variables used to describe a multipermutation is larger than necessary. Since multipermutations consist of many replicated entries, the information that describe the relative positions among these replicated entries are redundant. This suggests that we can reduce the number of variables used to specify multipermutations. 

In order to elaborate the second issue, we briefly review some concepts. In~\cite{wadayama2012lpdecodable}, a codebook is obtained by permuting an initial vector $\bfs$ with a set of permutation matrices. 
If $\bfs$ contains duplicates, then there exists two different permutation matrices $\bfPi_1$ and $\bfPi_2$ such that $\bfs\bfPi_1 = \bfs\bfPi_2$.  This means that we cannot differentiate between $\bfs\bfPi_1$ and $\bfs\bfPi_2$ by the matrices $\bfPi_1$ and $\bfPi_2$. The consequence is that minimizing the Hamming distance of two multipermutations is not equivalent to minimizing the Hamming distance between two permutation matrices\footnote{This is defined as the number of disagreeing entries (cf. Section~\ref{section.mpm}).}. 

In this paper, we introduce the concept of multipermutation matrices to represent multipermutations which, we will see, will address the above two problems.
Multipermutation matrices and multipermutations have a one-to-one relationship. Using multipermutation matrices to represent multipermutations reduces the number of variables needed to characterize a multipermutation. Further, due to the one-to-one relationship, minimizing the Hamming distance of two multipermutations is equivalent to minimizing the Hamming distance between the two corresponding multipermutation matrices. In order to construct codes that can be decoded using LP decoding, we develop a simple characterization of the convex hull of multipermutation matrices. The characterization is analogous to the well known Birkhoff polytope. These tools we introduce are the basis for all the code constructions that follow. They may also be of independent interests to the optimization community. 

In Section~\ref{section.lpd_MP_code}, we use the characterization of multipermutation matrices to define LP-decodable multipermutation codes. We explore several connections between LP-decodable multipermutation codes and the LP-decodable permutation codes proposed in~\cite{wadayama2012lpdecodable}. In particular, we show that both frameworks are sufficient to describe any multipermutation code. In other words, the interesting question is whether there are good codes (in terms of rate and error performance) that can be described efficiently. We show an easy description of a code construction from~\cite{shieh2010decoding}, which has known rate and distance properties.

Our third set of contributions is our LP decoding formulations (Section~\ref{section.lpdecoding}). First, we derive LP decoding algorithms for arbitrary memoryless channels. For the AWGN channel, our formulation is equivalent to the one proposed in~\cite{wadayama2012lpdecodable}. For the discrete memoryless $q$-ary symmetric channel, the LP decoding objective is equivalent to minimizing the Hamming distance. Second, we derive an LP decoding algorithm that minimizes the Chebyshev distance. We show some preliminary numerical results for the code from~\cite{shieh2010decoding} that is described in Section~\ref{section.lpd_MP_code}.

\section{Preliminaries}
\label{section.preliminary}
In this section, we briefly review the concept of permutation matrices and the code construction approach proposed by Wadayama and Hagiwara in~\cite{wadayama2012lpdecodable}. 

It is well known that every permutation from the symmetric group $\symgroup_n$ corresponds to a unique $n\times n$ permutation matrix.
A permutation matrix is a binary matrix such that every row or column sums up to $1$. In this paper, all permutations and multipermutations are row vectors. Thus, if $\bfPi$ is the permutation matrix for a permutation $\pi$, then $\pi = 
\imath \bfPi$ where $\imath = (1,2,\dots,n)$ is the identity permutation. We let $\Pi_n$ denote the set of all permutation matrices of size $n\times n$.

\begin{definition}
\label{def.linearly_constrained_permutation_matrix}
(cf.~\cite{wadayama2012lpdecodable})
Let $K$ and $n$ be positive integers. Assume that $\bfA \in \mathbb{Z}^{
K\times n^2}$, $\bfb \in \mathbb{Z}^K$, and let ``$\trianglelefteq
$'' represent ``$\leq$'' or ``$=$''. A set of linearly constrained permutation matrices is defined by 
\begin{equation}
\Pi(\bfA,\bfb,\trianglelefteq
) := \{\bfPi \in \Pi_n | \bfA \vect(\bfPi) \trianglelefteq \bfb \},
\end{equation}
where $\vect(\cdot)$ is the operation of concatenating all columns of a matrix to form a column vector. 
\end{definition}

\begin{definition}
(cf.~\cite{wadayama2012lpdecodable})
\label{def.LP_decodable_perm_code}
Assume the same set up as in Definition~\ref{def.linearly_constrained_permutation_matrix}.
Suppose also that $\bfs \in \real^n$ is given. The set of vectors $\Lambda(\bfA, \bfb, \trianglelefteq, \bfs)$ given by 
\begin{equation}
\label{eq.lcpc}
\Lambda(\bfA,\bfb,\trianglelefteq
,\bfs) := \{\bfs \bfPi  \in \real^n | \bfPi \in \Pi(\bfA, \bfb, \trianglelefteq)\}
\end{equation}
is called an LP-decodable permutation code. $\bfs$ is called an ``initial vector''.\footnote{In this paper, we always let the initial vector be a row vector. Then $\bfs \bfPi$ is again a row vector. This is different from the notations followed by~\cite{wadayama2012lpdecodable}, where the authors consider column vectors.}
\end{definition}

Note the vector $\bfs$ may contain duplicates as follows
\begin{equation}
\label{eq.repeating_d}
\bfs = (\underbrace{t_1,t_1,\dots,t_1}_{r_1},\underbrace{t_2,t_2,\dots,t_2}_{r_2},\dots,\underbrace{t_m,t_m,\dots,t_m}_{r_m}),
\end{equation}
where $t_i \neq t_j$ for $i\neq j$ and there are $r_i$ entries with value $t_i$. In this paper, we denote by $\bfr = (r_1,\dots,r_m)$ the multiplicity vector. Note that $\sum_{i = 1}^m r_i = n$. 
Let $\bft := (t_1,t_2,\dots,t_m)$. Further, we let $\mcI_i := \{\sum_{l= 1}^{i-1} r_l + 1,\sum_{l= 1}^{i-1} r_l + 2,\dots,\sum_{l = 1}^{i-1} r_l + r_i\}$ for all $i = 1,\dots,m$. These are index sets such that $s_j = t_i $ for all $j \in \mcI_i$. We will use these sets several times throughout the paper.

At this point, it is easy to observe that the vector $\bfs$ can be uniquely determined by $\bft$ and $\bfr$. In the following section, we introduce the definition of multipermutation matrices that are parameterized by $\bfr$. Each multipermutation matrix corresponds to a permutation of $\bfs$. We will use $\bft$ to represent a vector with distinct entries throughout the paper in order to keep consistency. 

\section{Multipermutation matrices}
\label{section.mpm}

In this section, we introduce the concept of multipermutation matrices. Although as in~\eqref{eq.lcpc}, we can obtain a multipermutation of length $n$ by multiplying a permutation matrix of size $n\times n$ and an initial vector of length $n$, this mapping is not one-to-one. Thus $|\Lambda(\bfA,\bfb, \trianglelefteq, \bfs)| \leq |\Pi(\bfA,\bfb,\trianglelefteq)|$, where the inequality can be strict if there is at least one $i$ such that $r_i \geq 2 1$. As a motivating example, let a multipermutation be $\mpx = (1,2,1,2)$. Consider the following two permutation matrices
$$\bfPi_1 = \begin{pmatrix}
1& 0 & 0 & 0\\
0& 0 & 1 & 0\\
0& 1 & 0 & 0\\
0& 0 & 0 & 1
\end{pmatrix}
\text{ and }
\bfPi_2 = \begin{pmatrix}
0& 0 & 1 & 0\\
1& 0 & 0 & 0\\
0& 0 & 0 & 1\\
0& 1 & 0 & 0
\end{pmatrix}.
$$

Then $\mpx = \bfs \bfPi_1 = \bfs \bfPi_2$, where $\bfs = (1,1,2,2)$. In fact there are a total of four matrices that can produce $\mpx$. 

We first define multipermutation matrices. Given a fixed multiplicity vector, each multipermutation matrix corresponds to a unique multipermutation. Then, we discuss their connections to permutation matrices. Finally, we show a theorem that characterizes the convex hull of multipermutation matrices, a theorem that is crucial for our code constructions. 

\subsection{Introducing multipermutation matrices}
A multipermutation can be thought of as a ``ranking'' that allows ``draws''. Formally, let $\bfr$ be a multiplicity vector of length $m$ and let $n := \sum_{i = 1}^m r_i$. Consider the following multiset parameterized by $\bfr$ and $m$ 
\begin{equation*}
\{\underbrace{1,1,\dots,1}_{r_1},\underbrace{2,2,\dots,2}_{r_2},\dots,\underbrace{m,m,\dots,m}_{r_m}\}.
\end{equation*}
Then a \emph{multipermutation} is a permutation of this multiset. Note that values of $m$ and $n$ can be calculated from the vector $\bfr$. Thus we denote by $\mpset(\bfr)$ the set of all multipermutations with multiplicity vector $\bfr$.

\begin{definition}
\label{def.multipermutation_matrix}
Given a multiplicity vector $\bfr$ of length $m$ and $n = \sum_{i = 1}^m r_i$, we call a $m \times n$ binary matrix $\mpmX$ a \emph{multipermutation matrix} parameterized by $\bfr$ if $\sum_{i = 1}^m \mpmXent_{ij} = 1$ for all $j$ and $\sum_{j = 1}^n \mpmXent_{ij} = r_i$ for all $i$. Denote by $\mpmset(\bfr)$ the set of all multipermutation matrices parameterized by $\bfr$. 
\end{definition}

With this definition, it is easy to build a bijective mapping between multipermutations and multipermutation matrices. When the initial vector is $(t_1,\dots,t_m)$, the mapping $\mpset(\bfr)\mapsto\mpmset(\bfr)$ can be defined as follows: Let $\mpx$ denote a multipermutation, then it is uniquely represented by the multipermutation matrix $\mpmX$ such that $\mpmXent_{ij} = 1$ if and only if $\mpxent_j = t_i$. Conversely, to obtain the multipermutation $\mpx$, one only need to multiply the vector $(1,\dots,m)$ by $\mpmX$.

\begin{example}
Let the multiplicity vector be $\bfr = (2,3,2,3)$, let $\bft = (1,2,3,4)$ and let $$\mpx = (2, 1, 4, 1, 2, 3, 4, 4, 2, 3).$$ Then the corresponding multipermutation matrix is $$\mpmX = 
\begin{pmatrix}
0 & 1 & 0 & 1 & 0 & 0 & 0 & 0 & 0 & 0\\
1 & 0 & 0 & 0 & 1 & 0 & 0 & 0 & 1 & 0\\
0 & 0 & 0 & 0 & 0 & 1 & 0 & 0 & 0 & 1\\
0 & 0 & 1 & 0 & 0 & 0 & 1 & 1 & 0 & 0
\end{pmatrix}.
$$
The row sums of $\mpmX$ are $(2,3,2,3)$ respectively. Further, $\mpx = \bft \mpmX$.
\end{example}

$\bft$ is called the ``initial vector''. More generally, $\bft$ could depend on the physical modulation technique and does not have to be the vector $(1,\dots,m)$. For this generic setting, we have the following lemma.
\begin{lemma}
\label{lemma.multiperm_one_to_one}
Let $\bft$ be an initial vector of length $m$ with $m$ distinct entries. Let $\mpmX$ and $\mpmY$ be two multipermutation matrices parameterized by a multiplicity vector $\bfr$. Further, let $\mpx = \bft \mpmX$ and $\mpy = \bft \mpmY$. Then $\mpx = \mpy$ if and only if $\mpmX = \mpmY$.
\end{lemma}
\begin{proof}
First, it is obvious that if $\mpmX = \mpmY$ then $\mpx = \mpy$. Next, we show that if $\mpmX \neq \mpmY$ then $\mpx \neq \mpy$. We prove by contradiction. Assume that there exists two multipermutation matrices $\mpmX$ and $\mpmY$ such that $\mpmX \neq \mpmY$ and $\mpx = \mpy$. Then
\begin{align*}
\mpx - \mpy = \bft (\mpmX - \mpmY).
\end{align*}
Since $\mpmX \neq \mpmY$, there exists at least one column $j$ such that $\mpmX$ has a $1$ at the $k$-th row and $\mpmY$ has a $1$ at the $l$-th row where $k \neq l$. Then the $j$-th entry of $\bft (\mpmX - \mpmY)$ would be $t_k - t_l \neq 0$ due to the fact that all entries of $\bft$ are distinct. This contradict the assumption that $\mpx - \mpy = \bfzero$.
\end{proof}

At this point, one may wonder why this one-to-one relationship matters. We now discuss two important aspects in which multipermutation matrices are beneficial. 

\subsubsection{The Hamming distance}
The Hamming distance between two multipermutations is defined as the number of entries in which the two vectors differ from each other. More formally, let $\mpx$ and $\mpy$ be two multipermutations, then $\hamming(\mpx, \mpy) = |\{i| \mpxent_i \neq \mpyent_i\}|$. Due to Lemma~\ref{lemma.multiperm_one_to_one}, we can express the Hamming distance between two multipermutations using their corresponding multipermutation matrices. 
\begin{lemma}
\label{lemma.Hamming}
Let $\mpmX$ and $\mpmY$ be two multipermutation matrices. With a small abuse of notations, denote by $\hamming(\mpmX, \mpmY)$ the Hamming distance between the two matrices, which is defined by $\hamming(\mpmX, \mpmY) := |\{(i,j)| \mpmXent_{ij} \neq \mpmYent_{ij}\}|$. Then 
$$\hamming(\mpmX, \mpmY) = 2 \hamming(\mpx, \mpy),$$
where $\mpx = \bft \mpmX$ and $\mpy = \bft \mpmY$; recall that $\bft$ is an initial vector with distinct entries. Furthermore,
$$\hamming(\mpmX, \mpmY) = \tr(\mpmX^T (\bfE - \mpmY)),$$
where $\tr(\cdot)$ represents the trace of the matrix and $\bfE$ is an $n\times n$ matrix with all entries equal to $1$. \footnote{We note that $\tr(\bfX^T \bfY) = \sum_{i,j} X_{ij}Y_{ij}$ is the Frobenius inner product of two equal-sized matrices.}
\end{lemma}
\begin{proof}
For each entry $i$ such that $\mpxent_i \neq \mpyent_i$, the $i$-th column of $\mpmX$ differs from the $i$-th column of $\mpmY$ by two entries. As a result, the distance between multipermutation matrices is double the distance between the corresponding multipermutations.

Next, 
$\tr(\mpmX^T (\bfE - \mpmY)) = \sum_{ij} \mpmXent_{ij} (1 - \mpmYent_{ij})$. 
If $\mpmXent_{ij} = \mpmYent_{ij}$ then $\mpmXent_{ij} (1 - \mpmYent_{ij}) = 0$. Otherwise $\mpmXent_{ij} (1 - \mpmYent_{ij}) = 1$. Therefore $\hamming(\mpmX, \mpmY) = \tr(\mpmX^T (\bfE - \mpmY))$.
\end{proof}

We note that if we were to represent a multipermutation using a permutation matrix and an initial vector that contains duplicates, then we cannot get a direct relationship between the Hamming distance of two multipermutations and the Hamming distance of their permutation matrices. This is the second issue with the representation of~\cite{wadayama2012lpdecodable} discussed in the introduction.

\begin{example}
\label{example.hamming}
Let the initial vector be $\bfs = (1,1,2,2)$. Consider two permutation matrices $\bfPi_1$ and $\bfPi_2$, where $\bfPi_1$ is the identity matrix and $$\bfPi_2 = \begin{pmatrix}
0 & 1 & 0 & 0\\
1 & 0 & 0 & 0\\
0 & 0 & 0 & 1\\
0 & 0 & 1 & 0
\end{pmatrix}.$$ Then $\hamming(\bfs \bfPi_1, \bfs \bfPi_2) = 0$, however $d_{H}(\bfPi_1, \bfPi_2) = 8$. 

Alternatively, we can use $\bft = (1,2)$ as the initial vector and use multipermutation matrices. Then there is a unique $\mpmX = \begin{pmatrix}
1 & 1 & 0 & 0\\
0 & 0 & 1 & 1
\end{pmatrix}$ such that $\bfs = \bft \mpmX$.
\end{example}

We make an important observation from Example~\ref{example.hamming}: The mapping from permutation matrix to multipermutation is not one-to-one. In Example~\ref{example.hamming}, both $\bfPi_1$ and $\bfPi_2$ are mapped to the same vector. This is because $\bfs$ contains duplicates. 

\subsubsection{Reduction on the number of variables}
Another advantage of using multipermutation matrices is that they require fewer variables to describe multipermutations (cf. the first issue discussed in the introduction). The multipermutation matrix corresponding to a length-$n$ multipermutation has size $m\times n$, where $m$ is the number of distinct values in the multipermutation. 

This benefit can be significant when the multiplicities are large, i.e., $m$ is much smaller than $n$. For example, a triple level cell (TLC) flash memory has $8$ states per cell. If a multipermutation code has blocklength $1000$, then one needs an $8\times 1000$ multipermutation matrix to describe a codeword. The corresponding permutation matrix has size $1000\times 1000$.
\subsection{Geometry of multipermutation matrices}
In this section we prove an important theorem that characterizes the convex hull of all multipermutation matrices. We review the Birkhoff-von Neumann theorem for permutation matrices and then use it to prove our main theorem. 

We first review the definition of doubly stochastic matrices. We refer readers to~\cite{marshall2009inequalities} and references therein for more materials on doubly stochastic matrices.
\begin{definition}
An $n\times n$ matrix $\bfQ$ is doubly stochastic if 
\begin{itemize}
\item $Q_{ij} \geq 0$; 
\item $\sum_{i = 1}^n Q_{ij} = 1$ for all $j$ and $\sum_{j = 1}^n Q_{ij} = 1$ for all $i$.
\end{itemize}
The set of all doubly stochastic matrices is called the Birkhoff polytope.
\end{definition}

The set of all doubly stochastic matrices has a close relationship with the set of permutation matrices. Namely,
\begin{theorem}[Birkhoff-von Neumann Theorem, cf.~\cite{marshall2009inequalities}]
\label{theorem.birkhoff}
The permutation matrices constitute the extreme points of the set of doubly stochastic matrices. Moreover, the set of doubly stochastic matrices is the convex hull of the permutation matrices.
\end{theorem}

This theorem is the basis for the code construction method in~\cite{wadayama2012lpdecodable}. Namely, the LP relaxation for codes defined by Definition~\ref{def.LP_decodable_perm_code} is based on the Birkhoff polytope. In order to LP decoding algorithms using our definition of multipermutation matrices, we prove a similar theorem that characterizes the convex hull of multipermutation matrices. 

Denote by $\multipermpolytope(\bfr)$ the convex hull of all multipermutation matrices parameterized by $\bfr$, i.e. $\multipermpolytope(\bfr) = \conv(\mpmset(\bfr))$. $\multipermpolytope(\bfr)$ can be characterized by the following theorem.

\begin{theorem}
\label{theorem.multiperm_polytope}
Let $\bfr\in \integer_+^m$ and $\bfZ$ be a $m\times n$ matrix such that 
\begin{itemize}
\item[$(a)$] $\sum_{i = 1}^m Z_{ij} = 1$ for all $j = 1,\dots,n$.
\item[$(b)$] $\sum_{j = 1}^n Z_{ij} = r_i$ for all $i = 1,\dots,m$.
\item[$(c)$] $Z_{ij} \in [0,1]$ for all $i$ and $j$.
\end{itemize}
Then $\bfZ$ is a convex combination of multipermutation matrices parameterized by $\bfr$. Conversely, any convex combination of multipermutation matrices satisfies the conditions above.
\end{theorem}  
\begin{proof}
Consider a multipermutation $\mpx = (\mpxent_1, \dots, \mpxent_n)$ where each $\mpxent_k \in \{1,\dots,m\}$. Without loss of generality, we assume that $\mpx$ is in increasing order. Recall in Section~\ref{section.preliminary}, we let $\mcI_i = \{\sum_{l = 1}^{i-1}r_l +1 ,\dots, \sum_{l = 1}^{i}r_l\}$ be the index set for the $i$-th symbol. Then $\mpxent_k = i$ if $ k \in \mcI_i$. Let $\mpmX$ be the corresponding multipermutation matrix. Then $\mpmX$ has the following form
\begin{equation*}
\mpmX = \begin{pmatrix}
\underbrace{1 \; \dots\; 1}_{r_1} &  &&& \bigzero \\
& \underbrace{1 \; \dots\; 1}_{r_2} &&&\\
& & & \ddots &\\
\bigzero& & & & \underbrace{1 \; \dots\; 1}_{r_m}
\end{pmatrix},
\end{equation*}
where $X_{ik} = 1$ if $k \in \mcI_i$ and  $X_{ik} = 0$ otherwise.

Note that all multipermutation matrices parameterized by a fixed $\bfr$ are column permutations of each other. Of course, as already pointed out, not all permutations lead to distinct multipermutation matrices. To show that any $\bfZ$ satisfying $(a)$-$(c)$ is a convex combination of multipermutation matrices, we show that there exists a stochastic matrix $\bfQ$ such that $\bfZ = \mpmX \bfQ$. Then by Theorem~\ref{theorem.birkhoff}, $\bfQ$ can be expressed as a convex combination of permutation matrices. In other words, $\bfQ = \sum_{h}\alpha_h \bfPi_h$ where $\bfPi_h \in \Pi_n$ are permutation matrices; $\alpha_h \geq 0$ for all $h$ and $\sum_{h} \alpha_h = 1$. Then we have 
\begin{align*}
\bfZ = \mpmX  \sum_{h}\alpha_h \bfPi_h =  \sum_{h}\alpha_h (\mpmX \bfPi_h),
\end{align*}
where $\mpmX \bfPi_h$ is a column permuted version of the matrix $\mpmX$, which is a multipermutation matrix of multiplicity $\bfr$. This implies that $\bfZ$ is a convex combination of multipermutation matrices.

We construct the required $n \times n$ matrix $\bfQ$ in the following way. For each $i \in (1,2,\dots,m)$, let $\bfq^i$ be a length-$n$ \emph{column} vector, $q^i_j = \frac{1}{r_i} Z_{ij}$ for $j = 1,\dots,n$.  Then the $n\times n$ matrix
\begin{equation}
\bfQ^T := [\underbrace{\bfq^{1} | \bfq^{1} |\dots|}_{r_{1} \text{ of them}} \dots \underbrace{| \bfq^i | \bfq^i | \dots|}_{r_i \text{ of them}}
\dots \underbrace{|\bfq^{m} |\bfq^{m} \dots }_{r_{m} \text{ of them}}  ].
\end{equation}
In other words, $Q_{kj}  = \frac{1}{r_i} Z_{ij}$ for all $k \in \mcI_{i}$ and $j = 1,\dots,n$.
We now verify that $\bfZ = \mpmX \bfQ$  and that $\bfQ$ is doubly stochastic, which by our discussions above implies that $\bfZ$ is a convex combination of column permutations of $\mpmX$.
\begin{enumerate}
\item To verify $\bfZ = \mpmX \bfQ$, we need to show that $Z_{ij} = \sum_{k = 1}^n\mpmXent_{ik}Q_{kj}$. Since $\mpmX$ is a binary matrix,
\begin{align*}
\sum_{k= 1}^n \mpmXent_{ik}Q_{kj} = \sum_{k:\mpmXent_{ik} = 1} Q_{kj}.
\end{align*}
In addition, since $\mpx$ is sorted, $\mpmXent_{ik} = 1$ if and only if $k \in \mcI_{i}$.
By the definition of $\bfQ$, $Q_{kj}  = \frac{1}{r_i} Z_{ij}$ for all $k \in \mcI_{i}$. Therefore 
\begin{align*}
\sum_{k:\mpmXent_{ik} = 1} Q_{kj} = r_{i} \frac{Z_{ij}}{r_{i}} =Z_{ij}.
\end{align*}

\item Next we verify that $\bfQ$ is a double stochastic matrix. Since $0 \leq Z_{ij} \leq 1$ for all $i,j$,  $ Q_{ij} \geq 0$ for all $i,j$. By the definition of $\bfQ$, the sum of each row is $\|\bfq^i\|_1$ for some $i$. Thus $\|\bfq^i\|_1 = \sum_{j = 1}^n \frac{1}{r_i}Z_{ij} = 1$ by condition $(b)$. The sum of each column is 
\begin{align*}
\sum_{k = 1}^n Q_{kj} &=\! \sum_{i = 1}^m \sum_{k\in\mcI_i}  Q_{kj}\\
 &= \!\sum_{i = 1}^m \sum_{k\in\mcI_i} \frac{1}{r_i} Z_{ij} = \!\sum_{i = 1}^m Z_{ij} = 1,
\end{align*}
where the last equality is due to condition $(a)$.
\end{enumerate}
To summarize, for any given real matrix $\bfZ$ satisfying condition $(a)$-$(c)$ we can find a doubly stochastic matrix $\bfQ$ such that $\bfZ = \mpmX \bfQ$ for a particular multipermutation matrix $\mpmX$. This implies that $\bfZ$ is a convex combination of multipermutation matrices.

The converse is easy to verify by the definition of convex combinations and therefore is omitted.
\end{proof}

\section{LP-decodable multipermutation code}
\label{section.lpd_MP_code}
\subsection{Linearly constrained multipermutation matrices}
Using multipermutation matrices as defined in Definition~\ref{def.multipermutation_matrix}, we define the set of linearly constrained multipermutation matrices analogous to that in~\cite{wadayama2012lpdecodable}. Recall that we denote by $\mpmset(\bfr)$ the set of all multipermutation matrices parameterized by the multiplicity vector $\bfr = (r_1,\dots,r_m)$ and $n := \sum_{i = 1}^m r_i$.
\begin{definition}
\label{def.linearly_constrained_multipermutation_matrix}
Let $K$, $m$ and $n$ be positive integers. Assume that $\bfA \in \mathbb{Z}^{K \times (mn)}$, $\bfb \in \mathbb{Z}^K$, and let ``$\trianglelefteq
$'' represent ``$\leq$'' or ``$=$''. A set of linearly constrained multipermutation matrices is defined by 
\begin{equation}
\LCMM(\bfr,\bfA,\bfb,\trianglelefteq
) := \{\bfX \in \mpmset(\bfr) | \bfA \vect(\bfX) \trianglelefteq \bfb \},
\end{equation}
where $\bfr$ is the multiplicity vector.
\end{definition}

\begin{definition}
\label{def.LP_decodable_multiperm_code}
Let $K$, $m$ and $n$ be positive integers. Assume that $\bfA \in \mathbb{Z}^{K \times (mn)}$, $\bfb \in \mathbb{Z}^K$, and let ``$\trianglelefteq
$'' represent ``$\leq$'' or ``$=$''.
Suppose also that $\bft \in \real^m$ is given. The set of vectors $\LCMC(\bfr, \bfA, \bfb, \trianglelefteq, \bft)$ given by 
\begin{equation}
\label{eq.lcmc}
\LCMC(\bfr, \bfA,\bfb,\trianglelefteq
,\bft) := \{\bft \bfX \in \real^n | \bfX \in \LCMM(\bfr, \bfA, \bfb, \trianglelefteq)\}
\end{equation}
is called an LP-decodable multipermutation code.
\end{definition}

We can relax the integer constraints and form a \emph{code polytope}. Recall that $\multipermpolytope(\bfr)$ is the convex hull of all multipermutation matrices parameterized by $\bfr$.
\begin{definition}
\label{def.code_polytope}
The polytope $\mpcp(\bfr, \bfA,\bfb,\trianglelefteq)$ defined by
\begin{equation*}
\mpcp(\bfr, \bfA,\bfb,\trianglelefteq) := \multipermpolytope(\bfr) \bigcap \{\bfX \in \real^{m\times n} | \bfA \vect(\bfX) \trianglelefteq \bfb\}
\end{equation*}
is called the ``code polytope''. We note that $\mpcp(\bfr, \bfA,\bfb,\trianglelefteq)$ is a polytope because it is the intersection of two polytopes.
\end{definition}

We now briefly discuss these definitions and point out some key ingredients to notice.
\begin{itemize}
\item Definition~\ref{def.linearly_constrained_multipermutation_matrix} defines the set of multipermutation matrices. Due to Lemma~\ref{lemma.multiperm_one_to_one}, this set uniquely determines a set of multipermutations. The actual codeword that is transmitted (or stored in a memory system) is also determined by the initial vector $\bft$, where $\bft$ is determined by the modulation scheme used in the system. Definition~\ref{def.LP_decodable_multiperm_code} is  exactly the set of codewords after taking $\bft$ into account.
\item  Definition~\ref{def.code_polytope} is for decoding purposes. It will be discussed in detail in Section~\ref{section.lpdecoding}. As a preview, we note that we formulate an optimization problem with variables within the code polytope $\mpcp(\bfr, \bfA,\bfb,\trianglelefteq)$. In this optimization problem, the objective function is related to the initial vector $\bft$ but the constraints are not. Therefore $\mpcp(\bfr, \bfA,\bfb,\trianglelefteq)$ is not parameterized by $\bft$.
\item $\mpcp(\bfr, \bfA,\bfb,\trianglelefteq)$ is defined as the intersection between two polytopes. It is not defined as the convex hull of $\LCMC(\bfr, \bfA,\bfb,\trianglelefteq,\bft)$, which is usually harder to describe. However, this intersection may introduce fractional vertices, i.e. $\bfX \in \real^{m\times n}$ and $X_{ij} \in (0,1)$. Because of this, we call $\mpcp(\bfr, \bfA,\bfb,\trianglelefteq)$ a relaxation of $\conv(\LCMC(\bfr, \bfA,\bfb,\trianglelefteq,\bft))$. 
\end{itemize}
\subsection{Connections to LP-decodable permutation codes}
In this section, we discuss both LP-decodable permutation codes and LP-decodable multipermutation codes. We first show that both code construction frameworks are able to describe arbitrary codebooks. This means that neither definition reduces the code design space in terms of possible codebooks. In other words, they are able to achieve the capacity of multipermutation codes in any scenario when the optimal decoding scheme is used. In addition, we show that given an LP-decodable multipermutation code (defined by Definition~\ref{def.LP_decodable_multiperm_code}), there exist an LP-decodable permutation code (defined by Definition~\ref{def.LP_decodable_perm_code}) that produces the same codebook with the same number of linear constraints.

\begin{lemma}
\label{lemma.arbitrary_codebook}
Both the LP-decodable permutation code and the LP-decodable multipermutation code are able to describe arbitrary codebooks of multipermutations.
\end{lemma}
\begin{proof}
Denote by $\codebook$ the desired codebook. Then for any multipermutation $\mpy \notin \codebook$, we can construct a linear constraint per Definition~\ref{def.LP_decodable_perm_code} and Definition~\ref{def.LP_decodable_multiperm_code} such that only $\mpy$ is excluded by this constraint. 

We first consider multipermutation matrices. Without loss of generality, we assume that the initial vector $\bft$ under consideration has distinct entries. Therefore there exists a unique $\mpmY$ such that $\mpy = \bft \mpmY$. We now construct the following linear constraint
\begin{equation}
\tr(\mpmX(\bfE - \mpmY)) \geq 1.
\end{equation}
By Lemma~\ref{lemma.Hamming}, this constraint implies that the Hamming distance between $\mpy$ and $\mpx$ has to be greater than or equal to $1$. Adding this constraint excludes only $\mpy$ from the codebook. Because $\mpy$ is the only vector at distance $0$.

Now consider permutation matrices and assume that the initial vector $\bfs$ contains duplicates. We can use the same exclusion method to eliminate all permutation matrices $\bfPi$ such that $\mpy = \bfs\bfPi$. 

From the arguments above, we know that for any codebook, we can exclude all non-codewords by adding linear constraints. This implies that both the LP-decodable permutation code and the LP-decodable multipermutation code are able to describe an arbitrary codebook. 
\end{proof}

We note that Lemma~\ref{lemma.arbitrary_codebook} does not guarantee that the codebook can be expressed efficiently using linear constraints. In Proposition~\ref{prop.transform_mpc_to_pc}, we show that once we have an LP-decodable multipermutation code, we can obtain an LP-decodable permutation code with the same number of constraints. 
\begin{proposition}
\label{prop.transform_mpc_to_pc}
Let $\bft$ be a length-$m$ initial vector. Let $\bfr$ be the multiplicity vector for $\bft$. Then for any $\bfA$ and $\bfb$ that defines an LP-decodable multipermutation code $\LCMC(\bfr, \bfA,\bfb,\trianglelefteq
,\bft)$, cf.~\eqref{eq.lcmc}, there exist an $\bfA'$ such that $\LCMC(\bfr, \bfA,\bfb,\trianglelefteq
,\bft) = \Lambda(\bfA',\bfb,\trianglelefteq
,\bfs)$, where $\bfs$ is obtained by repeating each entry of $\bft$ $r_i$ times (i.e. Eq.~\eqref{eq.repeating_d}).
\end{proposition}
\begin{proof}
We first construct $\bfA'$ using $\bfA$. Then, we show that for every codeword in $\LCMC(\bfr, \bfA,\bfb,\trianglelefteq,\bft)$, it is also in $\Lambda(\bfA',\bfb \trianglelefteq,\bfs)$. Finally, we show the converse.

We construct $\bfA'$ row-by-row using $\bfA$. That is, row $l$ of $\bfA'$ is obtained from row $l$ of $\bfA$ for all $l$. We start by letting $\bfa$ be the first row of $\bfA$. Then the constraint induced by $\bfa$ is $\bfa \vect(\mpmX) \trianglelefteq b$, where $\mpmX$ is a multipermutation matrix and $b$ is the first value in $\bfb$. Since $\bfa$ is a length-$mn$ vector, we can relabel the entries of vector $\bfa$ by 
$$\bfa = (a_{11}, a_{21},\dots,a_{m1},a_{12},\dots,a_{m2},\dots,a_{1n},\dots,a_{mn})$$
Then $\bfa \vect(\mpmX) \trianglelefteq b$ can be rewritten as
\begin{equation}
\sum_{i \in \{1,\dots,m\},j \in \{1,\dots,n\}}a_{ij} \mpmXent_{ij} \trianglelefteq b
\end{equation}

Denote by $\bfa'$ the first row of $\bfA'$. We label entries of $\bfa'$ by 
$$\bfa' = (a'_{11}, a'_{21},\dots,a'_{n1},a'_{12},\dots,a'_{n2},\dots,a'_{1n},\dots,a'_{nn}).$$
Let $\mcI_i$ be the index set for symbol $i$. We construct $\bfa'$ by letting $a'_{kj} = a_{ij}$ for all $k \in \mcI_{i}$ and $j = 1,\dots,n$. In other words, the constraint $\bfa' \vect(\bfP) \trianglelefteq b$ is equivalent to
\begin{equation}
\label{eq.one_row_constraint}
\sum_{i \in \{1,\dots,m\},j \in \{1,\dots,n\}}a_{ij} \left(\sum_{k\in \mcI_{i}}P_{kj}\right) \trianglelefteq b
\end{equation}

For every row of $\bfA$, we repeat the above construction to obtain the corresponding row of $\bfA'$. 

Now we show that for each $\bfc \in \LCMC(\bfr, \bfA,\bfb,\trianglelefteq
,\bft)$, $\bfc \in \Lambda(\bfA',\bfb,\trianglelefteq
,\bfs)$, where $\bfA'$ is define by our definitions above. Since $\bfc \in \LCMC(\bfr, \bfA,\bfb,\trianglelefteq
,\bft)$, there exist a multipermutation matrix $\mpmX$ such that $\bfA \vect(\mpmX) \trianglelefteq \bfb$ and that $\bfc = \bft \mpmX$. We construct a permutation matrix as follows:
Let $\bfP$ be a $n\times n$ permutation matrix. We divide the rows of $\bfP$ into $m$ blocks $\bfB_{i}$, where each $\bfB_{i}$ is a $r_i \times n$ matrix that contains rows $k\in\mcI_i$ of $\bfP$. By Definition~\ref{def.multipermutation_matrix}, we know that the $i$-th row of $\mpmX$ contains exactly $r_i$ non-zero entries. We label these entries by $j_1,j_2,\dots,j_{r_i}$. In other words, $\mpmXent_{ij_{1}}= \cdots = \mpmXent_{ij_{r_i}} = 1$. Then, for all $l = 1,\dots,r_i$, we let the $l$-th row of $\bfB_{i}$ be a vector with a $1$ at the $j_l$-th entry. For example, if the $i$-th row of $\mpmX$ is $(1,0,1,0)$, then the $i$-th block of $\bfP$ is 
$$\bfB_i = \begin{pmatrix}
1 & 0 & 0 & 0\\
0 & 0 & 1 & 0\\
\end{pmatrix}.
$$

Then it is easy to verify that $\bfP$ is a permutation matrix, that $\bfP$ satisfies~\eqref{eq.one_row_constraint}, and that $\bfs \bfP = \bft \mpmX$. We omit the details here.

Finally, we show that for each $\bfc \in \Lambda(\bfA',\bfb,\trianglelefteq
,\bfs)$, $\bfc \in \LCMC(\bfr, \bfA,\bfb,\trianglelefteq
,\bft)$. Let $\bfP$ be the permutation matrix such that $\bfc = \bfs \bfP$ and that $\bfP\in \Pi(\bfA',\bfb,\trianglelefteq)$. We construct matrix $\mpmX$ by letting $\mpmXent_{ij} = \sum_{k\in \mcI_{i}}P_{kj}$. Then $\mpmX$ is the multipermutation matrix for vector $\bfc$ and $\mpmX \in \LCMM(\bfr, \bfA,\bfb,\trianglelefteq)$. Therefore $\bfc \in \LCMC(\bfr, \bfA,\bfb,\trianglelefteq, \bft)$.
\end{proof}

\subsection{Examples of LP-decodable multipermutation codes}
We provide two examples that leverage our code construction in Definition~\ref{def.LP_decodable_multiperm_code}. 
\begin{example}[Derangement] We say that a permutation $\pi$ is a derangement if $\pi_i \neq i$. For multipermutations, we can consider a generalized derangement defined as follows. Let $$\imath = (\underbrace{1, 1,\dots, 1}_{r_1},\underbrace{ 2, 2,\dots, 2}_{r_2},\dots,\underbrace{ m, m,\dots, m}_{r_m}).$$
Let $\mpx$ be a multipermutation obtained by permuting $\imath$. We say that $\mpx$ is a derangement if $\mpxent_i \neq \imath_i$ for all $i$.

In~\cite{wadayama2012lpdecodable}, the authors use Definition~\ref{def.LP_decodable_perm_code} to define such a set of permutations by letting $\tr(\bfP) = 0$, where $\bfP$ is a permutation matrix. We now extend this construction using Definition~\ref{def.LP_decodable_multiperm_code}. Using the same notations in Section~\ref{section.preliminary}, we let the linear constraints be
\begin{align*}
\sum_{j \in \mcI_i} X_{ij} = 0\text{ for all }i = 1,\dots,m.
\end{align*}
Suppose the initial vector $\bft = (1,2,\dots,m)$, then these constraints imply that symbol $i$ cannot appear at positions $\mcI_i$. For example, let $\bft = (1,2,3)$ and $\bfr = (2,2,2)$. Then the valid multipermutations are 
\begin{align*}
&(3, 3, 1, 1, 2, 2), (2, 2, 3, 3, 1, 1), (2, 3, 1, 3, 2, 1), \\
&(2, 3, 1, 3, 1, 2), (2, 3, 3, 1, 2, 1), (2, 3, 3, 1, 1, 2),\\
&(3, 2, 1, 3, 2, 1), (3, 2, 1, 3, 1, 2), (3, 2, 3, 1, 2, 1), \\
&(3, 2, 3, 1, 1, 2).
\end{align*}
\end{example}

\begin{example}
\label{example.code1}
In~\cite{shieh2010decoding}, the authors study multipermutation codes under the Chebyshev distance. The Chebyshev distance between two permutations (also multipermutations) $\mpx$ and $\mpy$ is defined as 
\begin{equation}
\chebyshev (\mpx, \mpy) = \max_{i} |\mpxent_i - \mpyent_i|.
\end{equation}
The following code construction is proposed in~\cite{shieh2010decoding}: 
\begin{definition}
\label{def.chebyshev_code}
Let $\bfr = (r,r,\dots,r)$ be a length-$m$ vector. Let $d$ be an integer such that $d$ divides $m$. We define
\begin{equation}
\codebook(r,m,d) = \{\bfx \in \mpset(\bfr)| \forall i \in \{1,\dots,mr\}, x_i \equiv i \bmod{d} \}.
\end{equation}
\end{definition}
This code has cardinality $(\frac{(ar)!}{(r!)^a})^d$ where $a = m/d$. It was shown in~\cite{shieh2010decoding} that the minimum Chebyshev distance of this code is $d$. Further, the rate of the code is observed to be relatively closer to a theoretical upper bound on the rate derived in~\cite{shieh2010decoding} when $r$ is larger. However no encoding or decoding algorithms are presented in~\cite{shieh2010decoding}.

In this paper, we first express the code construction using Definition~\ref{def.LP_decodable_multiperm_code}. Then, we derive a decoding algorithm in Section~\ref{section.lpdecoding}. While we have developed a tractable encoding algorithm, we do not present it due to space limitations and as it is not the main focus of this paper.

It is easy to verify that this code corresponds to the following linear constraints. 
\begin{itemize}
\item For all $j = 1,\dots,n$ and $i \nequiv j \bmod{d}$, $X_{ij} = 0$.  
\end{itemize}
As a concrete example, let $m = 6$, $r = 2$ and $d = 3$. Then the constraints are
\begin{align*}
X_{21} = X_{31} = X_{51} = X_{61} &= 0\\
X_{12} = X_{32} = X_{42} = X_{62} &= 0\\
\vdots& \\
X_{1, 12} = X_{2, 12} = X_{4, 12} = X_{5, 12} &= 0
\end{align*}
\end{example}

To summarize this section, we show how to construct codes using multipermutation matrices. Recall in Theorem~\ref{theorem.multiperm_polytope}, we characterize the convex hull of multipermutation matrices. We leverage this characterization in the next section to develop LP decoding algorithms.
\section{Channel model and LP decoding}
\label{section.lpdecoding}
In~\cite{wadayama2012lpdecodable}, the authors focused only on the AWGN channel. We first extend the LP decoding algorithm to arbitrary memoryless channels. The LP decoding objective is based on log-likelihood ratios, which is very similar to LP decoding of non-binary low-density parity-check codes proposed by Flanagan \emph{et al.} in~\cite{flanagan2009linearprogramming}. Next, we propose an LP decoding algorithm that minimizes the Chebyshev distance. 

Throughout this section, we will use $\bft$ to denote the initial vector. Without loss of generality, we assume that $\bft$ contains distinct entries. Thus the channel input space is $\cin = \{t_1,\dots,t_m\}$. 
\subsection{LP decoding for memoryless channels}
We first focus on memoryless channels. More formally, let $\cout$ be the output space of the channel. Let $\mpx$ be the transmitted multipermutation, which is a codeword from an LP-decodable multipermutation code. Let $\bfy$ be the received word. Then $\Pr_{\cout^n|\cin^n}(\bfy|\mpx) = \prod_{i = 1}^n \Pr_{\cout|\cin}(y_i|\mpxent_i)$. Under this assumption, we define a function $\LLRvect: \cout \mapsto \real^m$: $\LLRvect(y)$ is a length-$m$ row vector, each entry
$\LLRent_{i}(y)= \log \left(\frac{1}{\Pr(y|t_i)}\right)$. Further, we let $\LLRbig(\bfy) = (\LLRvect(y_1)^T|\dots|\LLRvect(y_n)^T)^T \in \real^{mn}$.

Then, maximum likelihood (ML) decoding can be written as
\begin{align*}
\hat{\mpx} &= \argmax_{\mpx \in \LCMC(\bfr,\bfA,\bfb,\trianglelefteq,\bft)} \Pr[\bfy|\mpx] \\
&= \argmax_{\mpx \in \LCMC(\bfr,\bfA,\bfb,\trianglelefteq,\bft)} \sum_{i = 1}^n \log \Pr[y_i|\mpxent_i]\\
&\eq^{\text{\scriptsize{(a)}}} \bft \left(\argmin_{\mpmX \in \LCMM(\bfr,\bfA,\bfb,\trianglelefteq)} \sum_{i = 1}^n \LLRvect(y_i) \mpmX^C_i\right) \\
&\eq^{\text{\scriptsize{(b)}}} \bft \left(\argmin_{\mpmX \in \LCMM(\bfr,\bfA,\bfb,\trianglelefteq)} \LLRbig(\bfy) \vect(\mpmX)\right),
\end{align*}
where $\mpmX^C_i$ is the $i$-th column of $\mpmX$. Equality (a) comes from the fact that for each $\mpx \in \LCMC(\bfr,\bfA,\bfb,\trianglelefteq,\bft)$ there exists an $\mpmX \in \LCMM(\bfr,\bfA,\bfb,\trianglelefteq)$ such that $\mpx = \bft \mpmX$. Further, since $\LLRvect(y_i) \mpmX^C_i =- \log \left(\Pr(y_i|t_i)\right)$, the maximization problem can be transformed to a minimization problem.
Equality (b) is simply a change to matrix notations. 

For this problem, we can relax the integer constraints $\LCMM(\bfr,\bfA,\bfb,\trianglelefteq)$ to linear constraints $\mpcp(\bfr,\bfA,\bfb,\trianglelefteq)$. Then the LP decoding problem is
\begin{equation}
\label{eq.LP_memoryless}
\begin{split}
\min \quad &  \LLRbig(\bfy) \vect(\bfX)\\ 
\st \quad & \bfX \in \mpcp(\bfr, \bfA,\bfb,\trianglelefteq)
\end{split}
\end{equation}

\begin{theorem}
The LP decoding problem~\eqref{eq.LP_memoryless} has an ML certificate. That is, whenever LP decoding outputs an integral solution, it is the ML solution.
\end{theorem}
\begin{proof}
Suppose that $\mpmX$ is the solution of the LP decoding problem and is integral. Then $\mpmX$ is a multipermutation matrix and $\bfA \vect(\mpmX) \trianglelefteq \bfb$. Therefore $\mpmX \in \LCMM(\bfr,\bfA,\bfb,\trianglelefteq)$. Since $\mpmX$ attains the maximum of the ML decoding objective, it is the ML solution.
\end{proof}
\subsubsection{The AWGN channel}
In the AWGN channel, $\Pr(y|t_i) = \frac{1}{\sqrt{2\pi}\sigma}e^{\frac{(y - t_i)^2}{2\sigma^2}}$, where $\sigma^2$ is the variance of the noise. Thus 
\begin{equation*}
\LLRbig(\bfy) = K(\underbrace{(y_1 - t_1)^2,\dots,(y_1 - t_m)^2}_{m}\\\dots \underbrace{(y_n - t_1)^2, \dots,(y_n - t_m)^2}_{m}),
\end{equation*}
 where $K >0$ is a scaling constant. Then $$\LLRbig(\bfy) \vect(\mpmX) = K\left(\sum_{i = 1}^n y_i^2 + \sum_{i = 1}^m r_i t_i^2 - 2 \bfu \vect(\mpmX)\right),$$ where $\bfu = (\underbrace{y_1 t_1,\dots, y_1 t_m}_{m}\dots\underbrace{y_n t_1, \dots, y_n t_m}_{m})$. Thus 
\begin{equation}
\argmin_{\bfX} \LLRbig(\bfy) \vect(\bfX) = \argmax_{\bfX} \tr((\bfy^T \bft)\bfX).
\end{equation}

We note that this formulation is the same as the LP decoding problem proposed in~\cite{wadayama2012lpdecodable}. We briefly restate the definition of pseudodistance in~\cite{wadayama2012lpdecodable} in the context of LP-decodable multipermutation codes and then state the upper bound for block error probability.
\begin{definition}
The function 
\begin{equation}
\omega_{\mathrm{AWGN}}(\mpmX, \hat{\mpmX}) =  \frac{\|\bft \mpmX \|_2^2 - \bft \hat{\mpmX} (\bft \mpmX)^T}{ \|\bft \hat{\mpmX} - \bft\mpmX\|_2}
\end{equation}
is called the pseudodistance where $\mpmX, \hat{\mpmX} \in \multipermpolytope(\bfr)$.
\end{definition}
\begin{proposition}
Let $\mathcal{V}$ be the set of fractional vertices for the code polytope $\mpcp(\bfr, \bfA,\bfb,\trianglelefteq)$. Suppose a codeword $\bft \mpmX$ is transmitted through the AWGN channel with noise variance $\sigma^2$. Then the block error probability is upper bounded by 
\begin{equation}
P_{\mathrm{error}} \leq \sum_{\hat{\mpmX} \notin \mathcal{V}\setminus \{\mpmX\}}Q\left(\frac{1}{\sigma}\omega_{\mathrm{AWGN}}(\mpmX, \hat{\mpmX})
\right),
\end{equation}
where $Q(\cdot)$ is the tail probability of the standard normal distribution.
\end{proposition}
\begin{proof}
We omit the proof because it is identical to the proof of Lemma~1 in~\cite{wadayama2012lpdecodable}.
\end{proof}

\subsubsection{Discrete memoryless $q$-ary symmetric channel}
Without loss of generality, we assume that the channel output space is the same as the input space. Namely, $\cin = \cout = \{1,\dots,m\}$. The transition probabilities are given by
\begin{equation*}
\Pr(y|x) = \begin{cases}
1 -p & \text{ if } y = x\\
\frac{p}{m - 1} & \text{ otherwise.}
\end{cases}
\end{equation*}
Let $\bfe(y)$ be a row vector such that $e(y)_i = 0$ if $i\neq y $ and $e(y)_i = 1$ if $i = y$. Further, we denote by $\bfY$ the matrix
\begin{equation*}
\bfY = [\bfe(y_1)^T | \bfe(y_2)^T | \dots | \bfe(y_n)^T].
\end{equation*}
Using these notations,
\begin{equation*}
\LLRvect(\bfy)  = \log\left(\frac{m-1}{p}\right) \bfone +  \log\left(\frac{1}{1-p}\cdot\frac{p}{m-1}\right)\bfe(y_i).
\end{equation*}
Then, 
\begin{equation*}
\LLRbig(\bfy) \vect(\mpmX) = \tr\left(\log\left(\frac{m-1}{p}\right)\bfE^T \mpmX \right.\\ \left. +\log\left(\frac{1}{1-p}\cdot\frac{p}{m-1}\right) \bfY^T \mpmX \right),
\end{equation*}
where $\bfE$ is an $m\times n$ matrix with all entries equal to one. Note that $\tr(\bfE^T \mpmX) = n$ is a constant and $\log\left(\frac{1}{1-p}\cdot\frac{p}{m-1}\right)$ is a negative constant. Therefore 
\begin{equation}
\argmin_{\bfX} \LLRbig(\bfy) \vect(\bfX) = \argmax_{\bfX} \tr(\bfY^T\bfX).
\end{equation}
Note that this is equivalent to minimizing the Hamming distance between $\bfX$ and $\bfY$ (cf. Lemma~\ref{lemma.Hamming}).

\subsection{LP decoding for the Chebyshev distance}
Permutations codes under the Chebyshev distance are proposed by Klove \emph{et al.} in~\cite{klove2010permutation} for flash memories. Its extension to multipermutations is introduced in~\cite{shieh2010decoding}. 
We first express the minimization of the Chebyshev distance using a linear program. We then provide a decoding example for the code construction in Example~\ref{example.code1}.

The decoding problem that minimizes the Chebyshev distance can be written as the following optimization problem:
\begin{align*}
\min \quad & \max_{i} |\mpxent_i - \mpyent_i| \\
\st \quad &  \mpx \in \LCMC(\bfr,\bfA,\bfb,\trianglelefteq,\bft)
\end{align*}
We can introduce a auxiliary variable $\delta$ and rewrite the problem as
\begin{align*}
\min \quad & \delta  \\ \st \quad &  \mpx \in \LCMC(\bfr,\bfA,\bfb,\trianglelefteq,\bft),\\
	& -\delta \leq \mpxent_i - \mpyent_i \leq \delta\text{ for all }i.
\end{align*}
Note that $\mpx = \bft \mpmX$, where $\mpmX \in \LCMM(\bfr,\bfA,\bfb,\trianglelefteq)$, and thus the problem can be reformulated as
\begin{align*}
\min \quad &  \delta  \\
\st \quad &  \mpmX \in \LCMM(\bfr,\bfA,\bfb,\trianglelefteq),\\
	& -\bftelta \leq \bft \mpmX - \mpy \leq \bftelta,
\end{align*}
where $\bftelta := (\delta, \delta,\dots,\delta)$ is a length-$n$ vector. Now we relax the problem to an LP
\begin{equation}
\label{eq.lp_chebyshev}
\begin{split}
\min \quad  & \delta  \\
\st \quad &  \mpmX \in \mpcp(\bfr, \bfA,\bfb,\trianglelefteq), \\ 
	& -\bftelta \leq \bft \mpmX - \mpy \leq \bftelta.
\end{split}
\end{equation}
We observe that the solution of LP decoding problem above is usually not unique. Therefore we adopt a simple rounding heuristic to obtain the final decoding result: Let $\hat{x}_j = \argmax_i X_{ij}$ for all $j = 1,\dots,n$. 

\begin{example}
We consider the code construction with the same parameters in Example~\ref{example.code1}. This code has minimum (Chebyshev) distance $3$. Therefore the optimal decoding scheme can correct $1$ error. We let the transmitted codeword be $\mpx = (1,2,3,4,5,6,1,2,3,4,5,6)$ and the received word be $\mpy = (2,1,4,3,6,5,2,1,4,3,6,5)$. The Chebyshev distance between $\mpx$ and $\mpy$ is $\chebyshev(\mpx, \mpy) = 1$. We also note that the Hamming distance is $\hamming(\mpx,\mpy) = 12$ and the Kendall tau distance is $\kendall(\mpx, \mpy) = 6$. 

We solve the LP problem~\eqref{eq.lp_chebyshev} using CVX~\cite{cvx}. The solution we obtain is $X_{11} = X_{33} = X_{44} = X_{66} = X_{17} = X_{39} = X_{4,10} = X_{6,12}= 0.5825$, $X_{41} = X_{63} = X_{14} = X_{36} = X_{47} = X_{69} = X_{1,10} = X_{3,12} = 0.4175$, $X_{22} = X_{55}= X_{28} = X_{5,11} = 1$ and the rest entries are zero. The minimum value for $\delta$ is $1$. Let $\hat{x}_j = \argmax_i X_{ij}$ for all $j = 1,\dots,n$. Then $\hat{\bfx}= (1,2,3,4,5,6,1,2,3,4,5,6)$.
\end{example}

\section{Conclusions}
In this paper, we develop several fundamental tools of a new multipermutation code  framework. We first propose representing multipermutations using binary matrices that we term multipermutation matrices. In order to apply LP decoding, we characterize the convex hull of multipermutation matrices. This characterization is analogous to the Birkhoff polytope of permutation matrices. Using this characterization, we formulate two LP decoding problems. The first LP decoding problem is based on minimizing the ML decoding objective. It applies to arbitrary memoryless channel. The second LP decoding problem is based on minimizing the Chebyshev distance. We demonstrate via an example that this LP decoding algorithm can be used to decode a code designed for the Chebyshev distance.


\end{document}